%% file: main.tex
\documentclass[11pt]{article}
\RequirePackage{algorithm,algorithmic,graphicx,amssymb,epsf,epic,url,complexity}
\RequirePackage{fullpage}
\usepackage{enumitem}
\usepackage{authblk}
\usepackage{mdframed}
\usepackage{microtype, complexity}
\usepackage[usenames,dvipsnames]{xcolor}
\usepackage{amsthm}
\usepackage{amsmath}
\usepackage{upref}
\usepackage{amsfonts}
\usepackage{graphicx}
\usepackage{enumitem}
\usepackage{latexsym}
\usepackage{amssymb}
\usepackage{amscd}
\usepackage{mdframed}
\usepackage[all]{xy}
\usepackage{comment}
\usepackage{xcolor}
\usepackage{amsthm}
\usepackage{framed}
\usepackage[margin=1in]{geometry}

\usepackage{mathrsfs}
\usepackage{amsfonts}
\usepackage{epsfig}
\usepackage{epstopdf}
\usepackage{tcolorbox}
\usepackage{tikz}
\usepackage[hang, small,labelfont=bf,up,textfont=it,up]{caption} 
\usepackage{authblk}

\newcommand {\ignore} [1] {}
\def\eps{\varepsilon}

\newcommand{\ceil}[1]{\lceil #1 \rceil}

\newcommand{\sets}{\mathcal{S}}
\newcommand{\univ}{\mathcal{U}} 
\newcommand{\cost}{\mathsf{cost}}
\newcommand{\brac}[1]{\left(#1\right)}
\newcommand{\lev}{\mathsf{lev}}
\newcommand{\plev}{\mathsf{plev}}

\newcommand{\asn}{\mathsf{asn}}
\newcommand{\cov}{\mathsf{cov}}

\newcommand{\ind}{\mathbf{1}}

\def\EMPH#1{\emph{\textcolor{BrickRed} {#1}}}

\usepackage[usenames,dvipsnames]{xcolor}
\usepackage{tcolorbox}
\usepackage{mdframed}

\def\EMPH#1{\emph{\textcolor{BrickRed} {#1}}}

\usepackage{comment}

\usepackage{hyperref}                     
\usepackage[capitalize]{cleveref}

\usepackage{aliascnt}

\newtheorem{thm}{Theorem}

\newaliascnt{lem}{thm}
\newtheorem{lem}[lem]{Lemma}
\aliascntresetthe{lem}

\newaliascnt{prop}{thm}

\aliascntresetthe{prop}

\newaliascnt{cl}{thm}
\newtheorem{cl}[cl]{Claim}
\aliascntresetthe{cl}

\newaliascnt{obs}{thm}
\newtheorem{obs}[obs]{Observation}
\aliascntresetthe{obs}

\newaliascnt{cor}{thm}
\newtheorem{cor}[cor]{Corollary}
\aliascntresetthe{cor}

\crefname{thm}{Theorem}{Theorems}
\Crefname{thm}{Theorem}{Theorems}

\crefname{lem}{Lemma}{Lemmas}
\Crefname{lem}{Lemma}{Lemmas}

\crefname{prop}{Proposition}{Propositions}
\Crefname{prop}{Proposition}{Propositions}

\crefname{cl}{Claim}{Claims}
\Crefname{cl}{Claim}{Claims}

\crefname{obs}{Observation}{Observations}
\Crefname{obs}{Observation}{Observations}

\crefname{cor}{Corollary}{Corollaries}
\Crefname{cor}{Corollary}{Corollaries}

\newtheorem{question}{Question}
\newtheorem{mdresult2}[question]{Question}

\title{Dynamic Set Cover with Worst-Case Recourse}

\author[1]{Shay Solomon\thanks{Funded by the European Union (ERC, DynOpt, 101043159). Views and opinions expressed are however those of the author(s) only and do not necessarily reflect those of the European Union or the European Research Council. Neither the European Union nor the granting authority can be held responsible for them. This research was also supported by the Israel Science Foundation (ISF) grant No.1991/1, and by a grant from the United States-Israel Binational Science Foundation (BSF), Jerusalem, Israel, and the United States National Science Foundation (NSF).}}

\author[1]{Amitai Uzrad \thanks{Funded by the European Union (ERC, DynOpt, 101043159). Views and opinions expressed are however those of the author(s) only and do not necessarily reflect those of the European Union or the European Research Council. Neither the European Union nor the granting authority can be held responsible for them. This research was also supported by the Israel Science Foundation (ISF) grant No.1991/1.}}

\affil[1]{Tel Aviv University}

\begin{document}

\date{\empty}

\begin{titlepage}
\def\thepage{}
\maketitle

\begin{abstract}
\noindent In the dynamic set cover (SC) problem, the input is a dynamic universe of at most $n$ elements and a fixed collection of $m$ sets, where each element belongs to at most $f$ sets and each set has cost in $[1/C,1]$. 
The objective is to \emph{efficiently} maintain an approximate minimum SC under element updates; efficiency is primarily measured by the \emph{update time}, but another important parameter is the \emph{recourse} (number of changes to solution per update).  
Ideally, one would like to achieve low \EMPH{worst-case} bounds on both update time and  recourse.

One can achieve approximation $(1+\epsilon)\ln n$ (greedy-based) or $(1+\epsilon)f$ (primal–dual-based) with  \emph{worst-case update time}   $O\brac{f\log n}$ (ignoring $\epsilon$-dependencies).
However, despite a large body of work,
 no algorithm with low update time (even amortized) and nontrivial \EMPH{worst-case recourse} is known,
even for unweighted instances ($C = 1$)! 

We remedy this by providing a transformation 
that, given as a \EMPH{black-box} a SC algorithm with approximation $\alpha$ and update time $T$, returns a set cover algorithm with approximation $(2 + \epsilon)\alpha$, update time $O(T +\alpha C)$ and worst-case recourse $O(\alpha C)$. Our main results are obtained by leveraging this transformation for constant $C$:

\begin{itemize}
\item 
For $f = O(\log n)$, applying the transformation on  the best primal-dual-based 
algorithm yields worst-case recourse $O(f)$. For constant $f$ (e.g., vertex cover), we get near-optimal bounds on all parameters. 
We also show  that for approximation $O(f)$---a recourse of $\Omega(f)$ is inevitable for primal-dual-based algorithms, even if amortized and  in the decremental setting.
\item 
For $f = \Omega(\log n)$, applying the transformation on the best greedy-based algorithm yields worst-case recourse $O(\log n)$. As our main technical contribution, we  show that by \EMPH{opening the black box} and exploiting a certain \EMPH{robustness} property of the greedy-based algorithm, the worst-case recourse can be reduced  to $O(1)$,  without sacrificing 
the other parameters, yielding 
a $((2 + \epsilon) \ln n)$-approximation with worst-case update time $O\brac{f\log n}$ and $O(1)$ worst-case recourse.
\end{itemize}

\end{abstract}

\end{titlepage}

\pagenumbering {arabic}

\input{Sec1}

\input{Sec2}

\input{Sec3}

\input{Sec4}

\newpage

\newcommand{\etalchar}[1]{$^{#1}$}


\end{document}

%% file: Sec1.tex
\section{Introduction}

The minimum set cover (SC) problem is among the most basic and extensively studied problems in combinatorial optimization. It represents the canonical covering problem, standing in duality to the family of packing problems. The problem is defined on a set system $(\univ, \sets)$, where $\univ$ is a universe of $n$ elements and $\sets$ is a family of $m$ subsets of $\univ$, each set $s \in \sets$ having cost $\cost(s) \in [\frac{1}{C},1]$ (in the {\em unweighted} setting $C=1$). The \emph{frequency} $f$ of the system is the maximum number of sets containing any element of $\univ$. A collection $\sets' \subseteq \sets$ is a SC if every element of $\univ$ is contained in at least one set in $\sets^\prime$, and the goal is to find a cover $\sets^*$ of minimum cost $\cost(\sets^*) = \sum_{s \in \sets^*} \cost(s)$. 

The classical {\em primal–dual} (PD) and {\em greedy} algorithms
achieve approximations $f$ and (roughly) $\ln n$, respectively. These guarantees are believed to be optimal: improving beyond $(f-\epsilon)$ is impossible assuming the Unique Games Conjecture \cite{khot2008vertex}, and beyond $(1-\epsilon)\ln n$ is NP-hard \cite{dinur2014analytical,williamson2011design}. Both the   PD and greedy algorithms can be interpreted through the lens of {\em linear programming} (LP). The PD algorithm explicitly maintains feasible primal and dual solutions, repeatedly raising dual variables for uncovered elements until some set constraint becomes tight, and that set is added to the primal solution,  terminating with an integral $f$-approximation.
The greedy algorithm admits a \emph{dual-fitting} analysis: it constructs a dual solution  by interpreting the incremental choices of the algorithm as assigning ``charges'' to covered elements. These dual values may initially violate feasibility, but scaling them down uniformly by the harmonic number $H(n)$ restores feasibility, yielding the familiar ($\approx \ln n$) approximation  through weak duality.
Despite its simplicity, the SC problem captures the fundamental structure underlying LP formulations of covering problems and serves as a canonical model for studying approximation and LP-relaxation techniques.

\paragraph{\textbf{Dynamic Set Cover.}}
There is a growing body of work on the SC problem in the \emph{dynamic} setting \cite{abboud2019dynamic,assadi2021fully,bhattacharya2017deterministic,bhattacharya2015design,bhattacharya2019new,bhattacharya2021dynamic,bukov2023nearly,gupta2017online,solomon2023dynamic,10756164}, where the universe $\univ$ undergoes updates (while $|\univ|\leq n$) and the family $\sets$ of $m$ sets is fixed. The algorithms are separated to the \emph{low-frequency regime} ($f = O(\log n)$) and the \emph{high-frequency regime} ($f = \Omega(\log n)$).
The objective is to maintain a near-optimal approximate 
SC {\em efficiently}, where the most well-studied measure of efficiency is the {\em update time}--- the time 
  required to update the maintained solution per update step. One may try to optimize the amortized (average) update time of the algorithm or its worst-case (maximum) update time.
 Another important efficiency measure, which is subject to growing research attention in recent years,  is the \EMPH{recourse}---the \emph{number of changes} to the maintained solution per update step. One may try to optimize the amortized as well as the worst-case recourse bounds, just as with the update time measure.

\paragraph{\textbf{Recourse.}}
Dynamic and online algorithms with low recourse have been studied 
for a wide range of optimization problems, including set cover, graph matching, MIS, Steiner tree, flow, and scheduling; see \cite{r2,r6,r9,r15,r16,r28,r30,r32,r35,r36,gupta2017online,r38,r39,r51,solomon2023dynamic} and the references therein. In applications such as resource allocation (e.g., assigning servers to clients in a network) and vehicle scheduling (e.g., delivery routing), replacing one solution component with another can be prohibitive---often far more costly than the computation time required to determine the replacement itself. This provides a primary motivation for algorithms with low recourse.
Another motivation is that when the recourse bound is small, 
all changes to the maintained solution can be efficiently reported after each update, which is important in practical implementations---particularly when the algorithm serves as a black-box subroutine within a larger data structure or algorithm.

\paragraph{\textbf{Amortized Versus Worst-case.}}
In many practical scenarios—particularly in systems required to provide real-time responses—tight control over the worst-case update time or recourse is essential; consequently, algorithms with only amortized guarantees may be inadequate in such settings.
Alas, for various dynamic problems, there is a strong separation between the state-of-the-art algorithms with low {\em amortized} versus {\em worst-case} bounds, in the update time or the recourse, or both.
We  discuss in detail this separation for the SC problem below. 
We note that a low worst-case update time does not imply a low worst-case recourse, with
the Maximal Independent Set (MIS) problem providing a prime example: \cite{8948632} and \cite{8948654} achieve $\mathsf{poly}(\log n)$ worst-case update time with high probability (against an oblivious adversary), while \cite{10.1145/3188745.3188922} proved a lower bound of  $\Omega(n)$ on the worst-case recourse.
However, there are problems, such as approximate matchings, for which algorithms with low update time were strengthened to also achieve low worst-case recourse \cite{ShayNoam}; more on this in \Cref{techsec}.

\paragraph{\textbf{Our Focus.}}
This work focuses on 
dynamic SC algorithms with  \EMPH{low worst-case} recourse. Despite a large body of work, all known algorithms with non-trivial worst-case \emph{recourse} incur slow update times.
Can one get low worst-case recourse together with fast update time? Fast does not mean sub-exponential time or even polynomial in the input size, but rather 
poly-log in the input size (and polynomial in $f$). The holy grail is to get the state-of-the-art worst-case update time together with low worst-case recourse.

\subsection{Prior Work}

\paragraph{\textbf{Slow Update Time, Worst-Case Recourse.}}
In the high-frequency regime, the only known low worst-case recourse algorithm, with approximation $O(\log n)$, takes exponential time \cite{gupta2017online}.
For low frequency and for unweighted instances, one can naively maintain a maximal matching (MM) with a worst-case update time of $O(n)$ and a worst-case recourse of $O(1)$, and this generalizes for hypergraphs, with the update time and recourse bounds increasing by a factor of $f^2$. Taking all matched vertices gives an $f$-approximate hypergraph vertex cover (i.e., $f$-approximate SC).

\paragraph{\textbf{Slow Update Time, Amortized Recourse.}} 
Two algorithms with $O(1)$ amortized recourse were given in \cite{gupta2017online} in the high-frequency and low-frequency regimes, achieving approximations  $O(\log n)$ and $O(f)$ respectively, but with super-linear update time; essentially that of recomputing from scratch each update step.

\paragraph{\textbf{Fast Amortized Time, Amortized Recourse.}}
\cite{solomon2023dynamic} presented a
$((1 + \epsilon) \ln n)$-approximation algorithm with $O(\min \{\log n , \log C \})$ amortized recourse and $O\brac{\frac{f\log n}{\epsilon^5}}$ amortized update time. It is also likely that an amortized recourse of $\mathsf{poly}(f)$ can be obtained with the algorithms of \cite{bhattacharya2019new} and follow-ups, but they did not analyze the amortized recourse. 
All these algorithms have slow worst-case update times.

\paragraph{\textbf{Fast Worst-case Time, High Recourse.}}
\cite{bhattacharya2021dynamic} gave a $((1 + \epsilon)f)$-approximation algorithm with $O\brac{\frac{f\log ^2 (Cn)}{\epsilon^3}}$ worst-case update time. 
The state-of-the-art algorithms in both the low and high-frequency regimes, with near-optimal approximations of $(1+\epsilon)f$ (PD-based) and $(1+\epsilon)\ln n$ (greedy-based) respectively, achieve a worst-case update time of $O\brac{\frac{f\log n}{\epsilon^2}}$ \cite{10756164}. All algorithms in \cite{bhattacharya2021dynamic,10756164}
incur a huge recourse, even amortized, since their approach relies on the notion of {\em schedulers}, which essentially means running multiple background solutions simultaneously, and switching between different solutions to the output frequently, blowing up the recourse.

The following fundamental question naturally arises:

\begin{tcolorbox} [width=\linewidth, sharp corners=all, colback=white!95!black]
\begin{question} \label{q1}
Can one obtain a SC algorithm with a nontrivial worst-case recourse bound, together with a fast (ideally worst-case) update time?
\end{question}
\end{tcolorbox}

\subsection{Our Contribution} 

We answer Question \ref{q1} in the affirmative,
in both the high-frequency and low-frequency regimes, for instances with small aspect ratio $C$. Our first result, proved in \Cref{techsec}, is the following:

\begin{thm} [Black-box Transformation] \label{warmupthm}
Let $\mathcal{ALG}$ be any algorithm that maintains an $\alpha$-approximate SC in  (amortized or worst-case) update time $T$. Using $\mathcal{ALG}$ as a black box, one can maintain a $((2+\epsilon) \alpha)$-approximate SC (assuming $\epsilon \leq 0.5$ and $\alpha \geq 2$) 
in  (amortized or worst-case) update time
$T + O\brac{\frac{\alpha \cdot C}{\epsilon}}$, and with a worst-case recourse of $O\brac{\frac{\alpha \cdot C}{\epsilon}}$.
\end{thm}

\noindent Feeding \Cref{warmupthm} the PD-based algorithm of \cite{10756164} directly yields:

\begin{cor} [Low-Frequency] \label{lf}
For any set system $(\univ, \sets)$ that undergoes a sequence of element insertions and deletions, where the frequency is always bounded by $f$, and for any $\epsilon \in (0, \frac{1}{4})$, there is a deterministic algorithm that maintains a $((2+\epsilon)f)$-approximate SC in $O\brac{\frac{f \cdot \log n}{\epsilon^2} + \frac{f \cdot C}{\epsilon}}$  worst-case update time and with $O\brac{\frac{f \cdot C}{\epsilon}}$ worst-case recourse.
\end{cor}

\noindent {\textbf{Remarks.}}
(1) \Cref{lf}
provides a recourse of $O(f)$ for constant $C$ (and $\epsilon$), and this is asymptotically optimal in the following sense: we show that any PD-based algorithm must have a recourse of $\Omega(f)$, even in the decremental setting. 
\noindent (2) For the minimum (weighted) {\em vertex cover} problem, by setting $f=2$ in \Cref{lf}, we can deterministically maintain
a $(4+\epsilon)$-approximate vertex cover in $O\brac{\frac{\log n}{\epsilon^2} + \frac{C}{\epsilon}}$  worst-case update time and with $O\brac{\frac{C}{\epsilon}}$ worst-case recourse. Despite a large body of work on dynamic vertex cover \cite{6108199,bhattacharya2017deterministic,bhattacharya2015design,bhattacharya2019new,bhattacharya2021dynamic,inbook123,6686191,10.1145/2700206,10.1145/1806689.1806753,doi:10.1137/1.9781611974331.ch51,7782946,Solomon18,10756164}, this is the first low worst-case recourse $O(1)$-approximation algorithm with low update time.

\bigskip

\noindent Feeding \Cref{warmupthm} the greedy-based algorithm of \cite{10756164} yields a worst-case recourse of $O\brac{\frac{\log n \cdot C}{\epsilon}}$. As our main technical contribution, we ``open the box'' to remove the $\log n$ factor from the recourse without harming any other parameter. We manage to do so by exploiting a certain \emph{robustness} property of the greedy-based algorithm.  
We say that a \emph{fixed} $\alpha$-approximate SC $X$ on the set system $(\univ,\sets)$ is \EMPH{robust} if following the deletion of up to $\delta \cdot \cost(X)$ {\em arbitrary} elements from $\univ$, for {\em any} $0 < \delta < 1$, $X$ is a $((1 + O(\delta)) \alpha)$-approximate SC. \footnote{We use $\delta$ and not $\epsilon$ since $\epsilon$ is reserved for the final approximation guarantee.} In \Cref{staticrob} we observe that, while the SC produced by the static PD algorithm is not robust, the static greedy algorithm produces a robust SC. A \emph{dynamic} SC algorithm is said to be robust if at any point in time the output SC is robust. As we will show in \Cref{dynamicrob}, although the notion of robustness is defined against a sequence of deletions, it naturally extends against any sequence of deletions and \emph{insertions}. We will prove there that the greedy-based dynamic algorithm of \cite{10756164} (with state-of-the-art worst-case update time) is robust (against deletions and insertions), which provides us with extra slack, as explained in \Cref{secim}, to prove the following:

\begin{thm} [High-Frequency] \label{hf}
For any set system $(\univ, \sets)$ that undergoes a sequence of element insertions and deletions, where the frequency is always bounded by $f$, and for any $\epsilon \in (0, \frac{1}{4})$, there is a deterministic algorithm that maintains a $((2+\epsilon)\ln n)$-approximate SC in $O\brac{\frac{f \cdot \log n}{\epsilon^2} + \frac{C}{\epsilon}}$  worst-case update time and with $O\brac{\frac{C}{\epsilon}}$ worst-case recourse.
\end{thm}

\noindent {\textbf{Remark.}} For constant $C$, and in particular unweighted instances ($C=1$), \Cref{hf} gives an \emph{asymptotically optimal} worst-case recourse, with the state-of-the-art worst-case update time, and with twice the near-optimal approximation.

\bigskip

\noindent As a direct corollary of \Cref{hf}, we obtain a result for the {\em minimum dominating set} (DS) problem. In the DS problem, we are given a graph $G = (V,E)$, where $n=|V|$, and each vertex has a cost assigned to it. The goal is to find a subset of vertices $D \subseteq V$ of minimum total cost, such that for any vertex $v \in V$, either $v \in D$ or $v$ has a neighbor in $D$. In the dynamic setting, the adversary inserts/deletes an edge upon each update step. We derive the result for the DS problem, improving previous results \cite{hjuler_et_al:LIPIcs.STACS.2019.35,solomon2023dynamic}, via a simple reduction to the SC problem (described in Section 6 in \cite{10756164}), which allows us to use our SC algorithm provided by \Cref{hf} as a black box.

\begin{thm} [Dominating Set] \label{ds}
For any graph $G = (V,E)$ that undergoes a sequence of edge insertions and deletions, where the degree is always bounded by $\Delta$, and for any $\epsilon \in (0, \frac{1}{4})$, there is a dynamic algorithm that maintains a $((2+\epsilon)\ln n)$-approximate minimum weighted dominating set in $O\brac{\frac{\Delta \cdot \log n}{\epsilon^2} + \frac{C}{\epsilon}}$ deterministic worst-case update time and with $O\brac{\frac{C}{\epsilon}}$ worst-case recourse.
\end{thm}

\subsection{Organization}
\Cref{robustsec} is devoted to the robustness property. We first (\Cref{staticrob}) observe that the SC produced by the static greedy algorithm is robust, whereas the one produced by the static PD algorithm is not. We then (\Cref{dynamicrob}) strengthen the observation on the robustness of the static greedy algorithm by proving that the greedy-based {\em dynamic} algorithm of \cite{10756164}, which achieves the current state-of-the-art worst-case update time, is also robust. \Cref{techsec} is devoted to the proof of \Cref{warmupthm}, which, when combined with the state-of-the-art PD-based algorithm of \cite{10756164}, directly implies \Cref{lf}. When \Cref{warmupthm} is combined with the state-of-the-art greedy-based algorithm of \cite{10756164}, a worst-case recourse of $O\brac{\frac{\log n \cdot C}{\epsilon}}$ is obtained; 
by carefully exploiting the robustness property proved in \Cref{dynamicrob},
in \Cref{secim} we demonstrate that the $\log n$ factor can be removed from this recourse bound to achieve a recourse of $O(\frac{C}{\epsilon})$, thus proving \Cref{hf}.

%% file: Sec2.tex
\section{The Robustness Property} \label{robustsec}

The main result of this section is  \Cref{robusthf},
which asserts that the greedy-based dynamic algorithm of \cite{10756164} is robust. This robustness property is crucial for shaving off the $\log n$ term from the recourse, thereby proving \Cref{hf} (see \Cref{secim} for details). We begin 
in \Cref{staticrob}
by observing (\Cref{obs:warm}) that the \emph{static} greedy algorithm produces a robust SC, and then note that the static PD algorithm is not robust. We view \Cref{obs:warm} as a warm-up for the main robustness lemma 
of the dynamic algorithm \cite{10756164} (\Cref{robusthf}), but it might also be of independent interest.

\subsection{The Robustness Property I: The Greedy and PD Static Algorithms} \label{staticrob}

\begin{obs}[Warm-up: Static greedy is robust] \label{obs:warm}
Let $(\univ,\sets)$ be any set system and let $X$ be a SC returned by the classic greedy algorithm, and $OPT$ the optimum value. Let $D \subseteq \univ$ be any set of deleted elements with $|D| \leq \delta \cdot \cost(X)$ for any $0 < \delta < 1$, and let $OPT'$ be the optimum value after these deletions.
Then: $$\frac{\cost(X)}{OPT'} \;\le\; \frac{H_n}{1-\delta} \;\leq\; (1+O(\delta))\,\ln n,$$

\noindent where $n:=|\univ|$ and $H_n$ is the $n$-th harmonic number.
\end{obs}

\begin{proof}
We analyze the greedy algorithm via LP duality. For weighted SC:
\[
\begin{aligned}
  \text{(P)}\quad & \min \sum_{S\in\mathcal{S}} \cost(S) \cdot x_S
  && \text{s.t. } \sum_{S\ni e} x_S \ge 1 && \forall e\in U, \quad x_S\ge 0, \\[3pt]
  \text{(D)}\quad & \max \sum_{e\in U} y_e
  && \text{s.t. } \sum_{e\in S} y_e \le \cost(S) && \forall S\in\mathcal{S}, \quad y_e\ge 0.
\end{aligned}
\]

\noindent In the $j$'th iteration, when the greedy algorithm selects a set $S_j$ that newly covers $c_j$ elements, assign to each such element $e$ a 
weight: $q(e) := \frac{\cost(S_j)}{c_j}$.
Greedy’s total cost 
equals the total weight assigned:
$\sum_{e\in \univ} q(e) = \cost(X)$.
Also, by the standard analysis, for any $S\in\mathcal{S}$: 
$$\sum_{e\in S} q(e) \le H_{|S|} \cdot \cost(S) \le H_n \cdot \cost(S),$$
\noindent so scaling $y_e := \frac{q(e)}{H_n}$ makes $y$ dual-feasible for (D).
By weak duality, $OPT \geq \sum_e y_e = \frac{\cost(X)}{H_n}$, giving Greedy’s $H_n$-approximation. Let $D\subseteq \univ$ be any set of deleted elements, and define the following {\em restricted weights} $q'(e)$ to be $0$ if $e \in D$ and $q(e)$ otherwise.
Each $q(e)\le 1$, hence the total removed weight satisfies $\sum_{e\in D} q(e) \leq |D| \leq \delta \cdot \cost(X)$.
Therefore, the non-deleted weight is:
\begin{equation}\label{eq:remain}
  \sum_{e \in \univ \setminus D} q'(e) = \cost(X) - \sum_{e\in D} q(e) \ge (1-\delta) \cdot \cost(X).
\end{equation}
For every set $S$, $\sum_{e\in S} q'(e)\le \sum_{e\in S} q(e)\le H_n \cdot \cost(S)$,
so 
$y'(e):= \frac{q'(e)}{H_n}$ is dual-feasible for the surviving instance $(\univ \setminus D,\mathcal{S})$. By weak duality,
\[
  OPT' \;\ge\; \sum_{e \in \univ \setminus D}  y'(e)
  = \frac{1}{H_n}\sum_{e \in \univ \setminus D} q'(e)
  \;\geq\; \frac{1-\delta}{H_n}\,\cost(X).
\]
Rearranging 
yields the claimed robustness bound:
\[
  \frac{\cost(X)}{OPT'} \;\le\; \frac{H_n}{1-\delta}
  \;\leq\; (1+O(\delta))\, \ln n.
\]
\end{proof}

\noindent {\textbf{Remark.}} Replacing $H_n$ with $H_d$, where $d=\max_{S\in\mathcal{S}}|S|$, gives:

\[
\displaystyle \frac{\cost(X)}{OPT'}\le\frac{H_d}{1-\delta} \leq (1+O(\delta)) \ln d.
\]

\paragraph{\bf The PD algorithm is not robust.} Consider the following canonic unweighted instance: For each element $e_i$, where $i = 1, \ldots ,n$, create $f$ sets $S^1_i, \ldots , S^f_i$ that contain only $e_i$. The PD algorithm could add all $n \cdot f$ sets to the SC, while $OPT$ is of size $n$.
If $\ell$ elements are deleted, for any $\ell$, we have $|OPT| = n-\ell$.
In particular, for  $\ell = n-O(1)$ (the extreme case is when $\ell = n$), $OPT$ is of size $O(1)$ (or 0); thus already for $\delta = 1/f$, after deleting a $\delta$ fraction of the elements in the solution produced by the PD algorithm, the approximation  may become arbitrarily large, thus the PD algorithm is not robust (and actually far from robust). 
Furthermore, if $\ell = n-O(1)$, we will need to delete $n \cdot f - A$ sets from the set cover to achieve an approximation factor of $O(A)$, which gives an amortized recourse of $\Omega(f)$ for any reasonable approximation.
This shows the asymptotic optimality of the recourse bound provided by \Cref{lf} (for constant $C$ and $\delta$): any PD-based algorithm must have a recourse of $\Omega(f)$, even amortized and even in the decremental setting.

{Although this example with singleton sets is degenerate,
as there is no reason to take more than one set to cover any element, 
one can use a less degenerate example with more dummy sets and elements.
The key point remains: if a single deletion can cause $|OPT|$ to drop, this does not mean that before the deletion the approximation was strictly better than $f$, as is with the greedy algorithm.
}

\subsection{The Robustness Property II: The Dynamic Algorithm 
\cite{10756164}} \label{dynamicrob}

In this section we prove that the greedy-based algorithm of \cite{10756164} is robust. In fact, we prove a stronger variant of this property, that the algorithm is also ``robust against insertions'' in a sense, as defined below.

\begin{lem} [\cite{10756164} is Robust - Greedy-Based] \label{robusthf}
Consider a solution given by \cite{10756164} denoted by $\mathcal{B}$, yielding an approximation factor of $(1 + \epsilon) \cdot \ln n$. Then throughout the next $\delta \cdot \cost(\mathcal{B})$ update steps (for any $0 \leq \delta \leq 1$) where all we do is maintain $\mathcal{B}$ naively (adding an arbitrary set containing each inserted element to the solution), the approximation factor is $(1 + O(\delta))(1 + \epsilon) \cdot \ln n$.
\end{lem}

\subsubsection{Preliminaries}

We first go over the main definitions of \cite{10756164}. Let $\beta = 1+\epsilon$. All sets $s \in \mathcal{S}$ are assigned a level value $\lev(s)\in [-1, L]$ where $L = \ceil{\log_\beta(Cn)} + \ceil{10\log_\beta 1/\epsilon}$. Throughout the algorithm, a valid set cover $\mathcal{B}\subseteq \sets$ is maintained for all elements. Each element $e\in \univ$ is assigned to one of the sets $s\in \mathcal{B}$, which we denote by $\asn(e)$, and conversely, for each set $s\in \sets$, define its {\em covering set} $\cov(s)$ to be the collection of elements in $s$ that are assigned to $s$, namely $\cov(s) = \{e\mid \asn(e) = s\}$. The level of an element $e$ is defined as the level of the set it is assigned to, namely $\lev(e) = \lev(\asn(e))$, and we are guaranteed that $\lev(e) = \max\{lev(s) \vert s \ni e\}$, meaning $e$ is assigned to the set with the highest level containing $e$. We define the level of each set $s \notin \mathcal{B}$ to be $-1$, whereas the level of each set $s \in \mathcal{B}$ will lie in  $[0,L]$, so in particular we will have $\lev(e)\in [0, L]$, for each element $e \in \mathcal{U}$. Let $S_i = \{s\mid \lev(s) = i \}, \forall i\in [-1, L]$, and $E_i = \{e\in \univ\mid \lev(e) = i\}, \forall i\in [0, L]$.

Besides the level value $\lev(e)$ for elements $e$, a value of \emph{passive level} $\plev(e)$ such that $\lev(e)\leq \plev(e) \leq L$ is also maintained, which plays a major role in the algorithm. In contrast to the level $\lev(e)$ of an element $e$, which may decrease (as well as increase) by the algorithm, its passive level $\plev(e)$ is monotonically non-decreasing throughout its lifespan.

An element is said to be \emph{dead} if it was deleted by the adversary, hence it is supposed to be deleted from $\univ$ --- but it currently resides in $\univ$ as the algorithm has not removed it yet. An element is said to be \emph{alive} if it is not dead. To avoid confusion, the notation $\univ^+\supseteq \univ$ is used to denote the set of all dead and alive elements (i.e., the elements in the view of the algorithm), while $\univ$ is the set of alive elements (i.e., the elements in the eye of the adversary). We next introduce the following key definitions.

\begin{def} \label{actpass}
For each level $k$, 
an element 
$e\in \univ^+$
is called \EMPH{$k$-active} (respectively, \EMPH{$k$-passive}) if 
$\lev(e) \leq k< \plev(e)$
(resp., $\plev(e)\leq k$)
and let $A_k = \{e\in \univ^+\mid \lev(e) \leq k< \plev(e)\}$ and $P_k = \{e\in \univ^+\mid \plev(e)\leq k\}$
be the sets of all $k$-active and $k$-passive elements, respectively.
Notice that $A_k \cup P_k$ is the collection of all elements at level $\leq k$, and $A_k \cap P_k = \emptyset$. Moreover, if $A_k \cap P_{j} \neq \emptyset$ for two levels $k \ne j$, then $k < j$. For each set $s\in \sets$, define $N_k(s) = A_k\cap s$.
\end{def}

\subsubsection{Invariants}

The algorithm of \cite{10756164} maintains the following three invariants. The first and second are identical to the ones in presented in \cite{10756164}, but the third here is a \emph{relaxed} version of the third there, which originally demanded that for each $k \in [0, L]$, we have $|P_k|\leq 2\epsilon \cdot|A_k|$. This relaxation is crucial for the robustness proof.

\begin{enumerate}

	\item For any set $s\in \sets$ and for each $k \in [0, L]$, we have $\frac{|N_k(s)|}{\cost(s)} < \beta^{k+1}$.

	\item For any set $s\in \mathcal{B}$, we have $\frac{|\cov(s)|}{\cost(s)}\geq \beta^{\lev(s)}$; we note that $\cov(s)$ may include dead elements, i.e., elements in $\univ^+ \setminus \univ$. In particular, $\lev(s)\leq \ceil{\log_\beta(Cn)}$. Moreover, for each $s\notin \mathcal{B}$, $\lev(s) = -1$.
		
	\item $\sum_{k=0}^{L-1}\brac{\beta^{-k} - \beta^{-k-1}} \cdot |P_k|\leq 2\epsilon \cdot \sum_{k=0}^{L-1} \brac{\beta^{-k} - \beta^{-k-1}} \cdot|A_k|.$

\end{enumerate}

\subsubsection{Proof of \Cref{robusthf}}    

For the most part, the proof of \Cref{robusthf} follows the proof of the approximation factor in \cite{10756164}, except in one key area. We will present the entire formal proof and state exactly where this deviation occurs. Specifically, the first and third parts will be similar to \cite{10756164}, whereas the second part will be new. 

\paragraph{Part I - Similar to \cite{10756164}.}
The left hand side of the third invariant becomes:
	$$\begin{aligned}
		\sum_{k = 0}^{L-1}(\beta^{-k} - \beta^{-k-1})\cdot |P_k| 
		&= \sum_{e\in \univ^+}\sum_{k = 0}^{L-1}(\beta^{-k} - \beta^{-k-1}) \cdot \ind[e\in P_k]\\
		&=\sum_{e\in \univ^+}\sum_{k = \plev(e)}^{L-1}(\beta^{-k} - \beta^{-k-1})\\
		& = \sum_{e\in \univ^+}\brac{\beta^{-\plev(e)}-\beta^{-L}}
	\end{aligned}$$
	and the right-hand side of the third invariant becomes:
	$$\begin{aligned}
		2\epsilon\sum_{k = 0}^{L-1}(\beta^{-k} - \beta^{-k-1})\cdot |A_k| 
		&= 2\epsilon\sum_{e\in \univ^+}\sum_{k = 0}^{L-1}(\beta^{-k} - \beta^{-k-1}) \cdot \ind[e\in A_k]\\
		&= 2\epsilon\sum_{e\in \univ^+}\sum_{k = \lev(e)}^{\plev(e)-1}(\beta^{-k} - \beta^{-k-1})\\
		& = 2\epsilon\sum_{e\in \univ^+}\brac{\beta^{-\lev(e)}-\beta^{-\plev(e)}},
	\end{aligned}$$
	which yields:	
	$$\sum_{e\in \univ^+}\brac{\beta^{-\plev(e)}-\beta^{-L}}\leq 2\epsilon\cdot \sum_{e\in \univ^+}\brac{\beta^{-\lev(e)} - \beta^{-\plev(e)}}$$
	or equivalently, by adding $\sum_{e\in \univ^+}\brac{\beta^{-\lev(e)} - \beta^{-\plev(e)}}$ on the both sides, 
	\begin{equation} \label{eq:basicub1}
	\sum_{e\in \univ^+}\brac{\beta^{-\lev(e)}-\beta^{-L}} ~\leq~ (1+2\epsilon)\cdot \sum_{e\in \univ^+}\brac{\beta^{-\lev(e)} - \beta^{-\plev(e)}}.
	\end{equation}

\noindent We emphasize the point that $\univ^+$ also includes dead elements. We also have that:
\begin{equation} \label{eq:basicub11}
\begin{aligned}
	\cost(\mathcal{B}) &= \sum_{s\in \mathcal{B}}\cost(s)\leq \sum_{s\in \mathcal{B}}\beta^{-\lev(s)}\cdot |\cov(s)| ~=~ \beta\cdot \sum_{s\in \mathcal{B}}\sum_{e\in \cov(s)\cap\univ^+}\beta^{-\lev(e)} 
	\\ &\le (1+O(\epsilon))\cdot \sum_{e\in\univ^+}\brac{\beta^{-\lev(e)} - \beta^{-L}},
\end{aligned}
\end{equation}
where the first inequality holds by the second invariant and the second holds as $\lev(e)\leq L/2$ and hence $\beta^{-\lev(e)} - \beta^{-L} \ge \beta^{-\lev(e)}(1- \beta^{-\ceil{10\log_\beta 1/\epsilon}}) \ge \beta^{-\lev(e)}(1-\eps)$. 
Combining \Cref{eq:basicub1} and \Cref{eq:basicub11} we obtain:
\begin{equation} \label{eq:basicub111}
	\cost(\mathcal{B}) ~\leq~ (1+O(\epsilon))\cdot \sum_{e\in \univ^+}\brac{\beta^{-\lev(e)} - \beta^{-\plev(e)}}.
	\end{equation}

\paragraph{Part II - New Ideas.}
Our goal will be to show that following $\delta \cdot \cost(\mathcal{B})$ updates we have:

\begin{equation} \label{eq:basicub1111}
	\cost(\mathcal{B}) ~\leq~ (1 + O(\delta))(1+O(\epsilon))\cdot \sum_{e\in \univ^+}\brac{\beta^{-\lev(e)} - \beta^{-\plev(e)}}.
	\end{equation}

\noindent The worst thing that can occur due to an insertion is that the cost of the maintained set cover grows by one, and the worst thing that can occur due to a deletion is that the passive level of the deleted element will drop from the highest level to level $0$, because this is what lowers the right hand side of \Cref{eq:basicub111} the most. We will first handle the deletions. 
By $\delta \cdot \cost(\mathcal{B})$ deletions, the right hand side of \Cref{eq:basicub111} can be lowered only by less than $(1+O(\epsilon)) \cdot \delta \cdot \cost(\mathcal{B})$. This is because for each deleted element $e$, $\beta^{-\plev(e)}$ can rise only by less than $1$, and the rest does not change. Thus, after the $\delta \cdot \cost(\mathcal{B})$ deletions, we can claim that:
\begin{equation} \label{eq:basicub3}
\cost(\mathcal{B}) ~\leq~ (1+O(\epsilon))\cdot \sum_{e\in \univ^+}\brac{\beta^{-\lev(e)} - \beta^{-\plev(e)}} + (1+O(\epsilon)) \cdot \delta \cdot \cost(\mathcal{B}).  
\end{equation}
\noindent Rearranging, and for any $\delta \leq 1$, we indeed get that \Cref{eq:basicub1111} holds. Regarding the insertions, in the worst-case $\cost(\mathcal{B})$ grew by a factor of $(1+\delta)$, since the cost of each set is no more than $1$. Clearly \Cref{eq:basicub1111} would still hold (by multiplying the right side of \Cref{eq:basicub1111} by $(1 + \delta)$).

\paragraph{Part III - Similar to \cite{10756164}.} Denote by $\sets^*$ an optimum SC. Next, let us lower bound $\cost(\sets^*)$ using the term $\sum_{e\in \univ^+}\brac{\beta^{-\lev(e)} - \beta^{-\plev(e)}}$. For any $s\in \sets^*$, consider the following three cases for any index $k\in [L]$: 
	\begin{itemize}
	\item $k < \log_\beta(1/\cost(s))-1$.
		
		By the first invariant, we have: $|N_k(s)| < \beta^{k+1}\cdot \cost(s) < 1$, so $|N_k(s)| = 0$.

		\item $\log_\beta (1/\cost(s))-1\leq k \leq \log_\beta (n/\cost(s))$.
		
		By the first invariant, we have: 
		
		$$\frac{1}{\epsilon}\brac{\beta^{-k} - \beta^{-k-1}}|N_k(s)| = \beta^{-k-1}|N_k(s)| < \cost(s).$$

		\item $k > \ceil{\log_\beta (n / \cost(s))} = k_0$.
		
		In this case, we use the trivial bound: $|N_k(s)| \leq n \leq \beta^{k_0}\cdot \cost(s)$, and so we have:
		$$\frac{1}{\epsilon}\brac{\beta^{-k} - \beta^{-k-1}}|N_k(s)| = \beta^{-k-1}|N_k(s)| \le \beta^{k_0 - k-1}\cdot \cost(s).$$
		
	\end{itemize}
Observe that:
	\begin{equation} \label{lefthand}
		\begin{aligned}
		\frac{1}{\epsilon} \sum_{k = 0}^{L-1}(\beta^{-k} - \beta^{-k-1})\cdot |N_k(s)| 
		&= \frac{1}{\epsilon}\sum_{e\in s}\sum_{k = 0}^{L-1}(\beta^{-k} - \beta^{-k-1}) \cdot \ind[e\in N_k(s)]\\
		&= \frac{1}{\epsilon}\sum_{e\in s}\sum_{k = \lev(e)}^{\plev(e)-1}(\beta^{-k} - \beta^{-k-1})\\
		& = \frac{1}{\epsilon}\sum_{e\in s}\brac{\beta^{-\lev(e)}-\beta^{-\plev(e)}}.
	\end{aligned}
	\end{equation}
	By the above case analysis, we have:
	\begin{equation}
	\begin{aligned} \label{righthand}
	\frac{1}{\epsilon} \sum_{k = 0}^{L-1}(\beta^{-k} - \beta^{-k-1})\cdot |N_k(s)|
	&= \frac{1}{\epsilon} \sum_{0 \le k < \log_\beta(1/\cost(s))-1}(\beta^{-k} - \beta^{-k-1})\cdot |N_k(s)| 
	\\&~~~+
	\frac{1}{\epsilon} \sum_{ \log_\beta(1/\cost(s))-1 \le k \le  \log_\beta (n/\cost(s))}(\beta^{-k} - \beta^{-k-1})\cdot |N_k(s)| 
	\\&~~~+
		\frac{1}{\epsilon} \sum_{\log_\beta (n/\cost(s)) < k \le L-1}(\beta^{-k} - \beta^{-k-1})\cdot |N_k(s)| 
	\\&<   \sum_{0 \le k < \log_\beta(1/\cost(s))-1} 0 
		\\&~~~+ 
  \sum_{ \log_\beta(1/\cost(s))-1 \le k \le  \log_\beta (n/\cost(s))} \cost(s) 
		\\&~~~+  \sum_{\log_\beta (n/\cost(s)) < k \le L-1}\beta^{k_0 - k - 1} \cost(s) 		
	\\&< \brac{0 + (\log_\beta(n) + 2) + 1/\epsilon}\cdot \cost(s).
	\end{aligned}
	\end{equation}
Combining \Cref{lefthand} with \Cref{righthand} yields
	$$\frac{1}{\epsilon}\sum_{e\in s}\brac{\beta^{-\lev(e)} - \beta^{-\plev(e)}}\leq \brac{\log_\beta(n) + 2 + 1/\epsilon}\cdot \cost(s).$$
	Therefore, as $\ln(1+\epsilon) = \epsilon + O(\epsilon^2)$, under the assumption that $\eps =  \Omega(1 / \log n)$ we have:
	$$\sum_{e\in s}\brac{\beta^{-\lev(e)} - \beta^{-\plev(e)}}\leq
	(1+O(\epsilon))\ln n\cdot \cost(s).$$
	Since $\sets^*$ is a valid set cover for all elements in $\univ$ (all the alive elements) and as for each dead element $e$ (in $\univ^+ \setminus \univ$) we have $\brac{\beta^{-\lev(e)} - \beta^{-\plev(e)}} = 0$, it follows that:
	\begin{equation} \label{eq:conclude} \sum_{e\in \univ^+}\brac{\beta^{-\lev(e)} - \beta^{-\plev(e)}} ~\leq~
	\sum_{s \in S^*} \sum_{e \in s} (\beta^{-\lev(e)} - \beta^{-\plev(e)})
	~\le~ (1+O(\epsilon))\ln n\cdot \cost(\sets^*).
	\end{equation}

\noindent We conclude that:
$$\begin{aligned}
	\cost(\mathcal{B}) & \leq (1 + O(\delta))(1+O(\epsilon))\cdot \sum_{e\in \univ^+}\brac{\beta^{-\lev(e)} - \beta^{-\plev(e)}} 
	\le (1 + O(\delta))(1+O(\epsilon))\ln n\cdot \cost(\sets^*),
\end{aligned}$$
where the first inequality holds by \Cref{eq:basicub1111} and the second by \Cref{eq:conclude}. By scaling $\epsilon$, we conclude the proof of \Cref{robusthf}.

%% file: Sec3.tex
\section{$O\brac{\frac{\alpha C}{\epsilon}}$ Worst-case Recourse: Black-box Transformation} \label{techsec}

This section is devoted to the proof of \Cref{warmupthm}. When combined with the PD-based algorithm of \cite{10756164}, it directly implies \Cref{lf}. The ideas developed in \Cref{robustsec} are not yet employed; they will come into play in \Cref{secim}, where we refine and strengthen the result obtained here.

\subsection{A Static Reconfiguration Problem} \label{sec:staticrecon}

The framework of {\em reconfiguration} problems is subject to extensive research attention; see \cite{Heuvel_2013,ito,nishimura,ShayNoam} and the references therein.  
The basic goal in reconfiguration problems is to compute a \emph{gradual} transformation between two feasible solutions, so that all intermediate solutions remain feasible. 
For the SC problem, given two feasible solutions $\sets'_1$ (the source) and $\sets'_2$ (the target), the goal is to {\em gradually} transform $\sets'_1$ into $\sets'_2$, while changing only a small number of sets at each step. Assuming both solutions are $\alpha$-approximations, 
it is trivial to get a gradual transformation where the approximation ratio throughout the process is at most $2\alpha$.
Moreover, such a factor 2 loss is inevitable, even for vertex cover; e.g., 
consider $K_{n/2,n/2}$ (the complete bipartite graph with $n/2$ vertices on each side), and assume the source and target solutions consist of the vertices on the left and right side, respectively. To maintain feasibility, we must add all right-side vertices to the solution before we can remove from it even one left-side vertex. 
On the other hand, for packing problems, and for approximate matchings in particular, a reconfiguration between any source and target matchings need not incur a factor $2+\epsilon$ blow-up in the approximation; indeed, \cite{ShayNoam} gave a matching reconfiguration algorithm that increases the approximation by   a  factor of
$1+\epsilon$, and used it to achieve a black-box transformation for dynamic matching algorithms with low worst-case recourse.

\subsection{Overview} \label{bbt}

Our goal is to transform a given dynamic SC algorithm, denoted $\mathcal{ALG}$, that may have high worst-case recourse into one with low worst-case recourse, via a black-box transformation. To achieve this, we \EMPH{reduce our dynamic problem}---of guaranteeing low worst-case recourse while preserving almost the same approximation and update time guarantees of the black-box $\mathcal{ALG}$---\EMPH{to a static reconfiguration problem}. The basic idea is to try and stick to the same output SC, changing it as little as possible to cope with element updates, until the approximation ratio deteriorates (a bit), at which stage we would need to switch the output SC to a better SC, but we must do it gradually rather than instantaneously!
To this end, we shall run $\mathcal{ALG}$ in the ``background'', and ``sample'' it as the output SC only \EMPH{once in a while}; we would like to gradually transform the current output SC (whose approximation has deteriorated) into the freshly sampled SC (whose approximation is good) throughout a sufficiently long time interval---as long as needed by the solution to the static reconfiguration problem.

Thus, we partition the update sequence into disjoint consecutive intervals, and at the start of each interval, the current output SC serves as the source solution, and the SC maintained by $\mathcal{ALG}$ serves as the target. We gradually transform the output from the source to the target during the interval; we would need to choose the interval lengths appropriately (see below), so that we will be able to complete the transformation process by the time the interval ends with a low worst-case recourse and with the required  approximation guarantee.
One hurdle in the dynamic setting that we completely ignored, which does not arise in the static reconfiguration problem, is that the element updates that occur during an interval may damage both the feasibility and approximation guarantee of the output SC. However, we show that the naive treatment of element updates for achieving feasibility has little effect on the approximation factor for small aspect ratio $C$, and in general one can simply increase the recourse by a factor of $C$.

\subsection{Algorithm}

Our algorithm will simulate $\mathcal{ALG}$ in the background. The update sequence will be divided into intervals. It begins with an ``initial interval'', and then the first interval, the second interval, etc. We will denote the output solution, the background solution given by $\mathcal{ALG}$, and an optimum solution at the beginning of the $i$th interval by $\mathcal{X}_i$, $\mathcal{B}_i$ and $\mathcal{OPT}_i$, respectively. We define the initial interval as the $0$th interval, and assume $\mathcal{B}_0 = \mathcal{X}_0$. During the $i$th interval (for $i \geq 1$), we want to gradually transform the output from $\mathcal{X}_i$ to $\mathcal{B}_i$. Since at the end of the $i$th interval $\mathcal{B}_i$ might no longer be a legal SC, for every element that is inserted during the interval, we add to the output a set containing it (this is not necessarily needed but in the worst-case this does occur so we assume so for simplicity). We define this as \emph{naively maintaining} or \emph{naively extending} the solution. Define the sets that are naively added to the output during the $i$th interval by $\mathcal{N}_i$. Thus, our goal is to gradually transform the output from $\mathcal{X}_i$ (source) to $(\mathcal{B}_i \cup \mathcal{N}_i)$ (target) by the end of the $i$th interval. 

To do so, we divide the $i$th interval (for $i \geq 1$) into two phases of equal length, the \emph{adding phase} and the \emph{removing phase}. In the adding phase we gradually add sets in $\mathcal{B}_i \setminus \mathcal{X}_i$ to the output, and in the removing phase we gradually remove sets in $\mathcal{X}_i \setminus \mathcal{B}_i$ from the output. In addition, as mentioned in the previous paragraph, during both phases we add naively one set per insertion. At the end of the second phase we have that $\mathcal{X}_{i+1} = (\mathcal{B}_i \cup \mathcal{N}_i)$. 

The length of the $i$th interval will be $\frac{\epsilon}{6 \alpha} \cdot \max \{\cost(\mathcal{X}_i),\cost(\mathcal{B}_i)\}$ update steps. So each phase will be of length $\frac{\epsilon}{12 \alpha} \cdot \max \{\cost(\mathcal{X}_i),\cost(\mathcal{B}_i)\}$ update steps. We emphasize that the $i$th interval can also simply be of length $\frac{\epsilon}{6 \alpha} \cdot \cost(\mathcal{X}_i)$ or $\frac{\epsilon}{6 \alpha} \cdot \cost(\mathcal{B}_i)$ update steps (regardless of which is larger), and the theorem will still hold, but for the sake of consistency with \Cref{secim} we take the maximum, which is necessary there. The initial interval will simply be of length $\frac{\epsilon}{6 \alpha} \cdot \cost(\mathcal{B}_0)$ where at the end of it $\mathcal{X}_1 = (\mathcal{B}_0 \cup \mathcal{N}_0)$.

\subsection{Analysis}

\begin{obs} [Recourse] \label{wobsrecourse}
The worst-case recourse is $\frac{12 \alpha C}{\epsilon} + 1$.
\end{obs}

\begin{proof}
The worst-case recourse during the initial interval is one. In each phase in the $i$th interval ($i \geq 1$) we gradually add/remove $\leq C \cdot \max \{\cost(\mathcal{X}_i),\cost(\mathcal{B}_i)\}$ sets. This is done in up to $\frac{\epsilon}{12 \alpha} \cdot \max \{\cost(\mathcal{X}_i),\cost(\mathcal{B}_i)\}$ update steps. We add one more set per insertion. 
\end{proof}

\begin{obs} [Update Time] \label{wobstime}
The worst-case update time is $T + O\brac{\frac{\alpha C}{\epsilon}}$.
\end{obs}

\begin{proof}
Clearly the bottleneck is either running $\mathcal{ALG}$ in the background or the recourse.
\end{proof}

\begin{obs} [Legal Solution] \label{wobslegal}
At all times the output is a legal SC solution (covers all elements).
\end{obs}

\begin{proof}
In the beginning of the $i$th interval we have two legal solutions, $\mathcal{X}_i$ and $\mathcal{B}_i$. Throughout the $i$th interval the output solution always fully contains at least one of them, plus a set containing each inserted element.
\end{proof}

\noindent For the approximation factor, we will prove by induction on the intervals the following claim: 

\begin{cl} [Approximation Factor] \label{wapprox} For any $i \geq 0$:
    \begin{enumerate}
        \item The approximation factor throughout the $i$th interval is $(2 + \epsilon) \cdot \alpha$.
        \item $\frac{\cost(\mathcal{X}_{i+1})}{\cost(\mathcal{OPT}_{i+1})} \leq (1 + \frac{\epsilon}{3}) \cdot \alpha$. 
        
    \end{enumerate}
\end{cl}

\subsubsection{Proof of \Cref{wapprox}}
In this section we prove \Cref{wapprox}. \Cref{warmupthm} would hold by \Cref{wobsrecourse}, \Cref{wobstime}, \Cref{wobslegal} and \Cref{wapprox} (1).
\noindent For any $i \geq 0$, denote the output solution, the background solution given by $\mathcal{ALG}$, and an optimum solution, $t$ update steps after the beginning of the $i$th interval until it ends, by $\mathcal{X}_i^t$, $\mathcal{B}_i^t$ and $\mathcal{OPT}_i^t$, respectively. For the highest possible $t$ (given a specific $i$) we have $\mathcal{X}_i^t = \mathcal{X}_{i+1}$, $\mathcal{B}_i^t = \mathcal{B}_{i+1}$ and $\mathcal{OPT}_i^t = \mathcal{OPT}_{i+1}$. We start by recording the following simple observation:

\begin{obs} \label{obsinsert}
In each update step the cost of the output can grow by up to $1$, since the cost of sets lies in the range $[\frac{1}{C}, 1]$. Similarly, in each update step the cost of the optimum can fall by up to $1$, since only one set can leave an optimum solution each update step (and its cost is upper bounded by $1$).  
\end{obs}

\noindent We begin by proving the base case.

\begin{obs} \label{w1}
For any $t$ we have that $\cost(\mathcal{X}_0^t) \leq \cost(\mathcal{X}_0) + \frac{\epsilon}{6} \cdot \cost(\mathcal{OPT}_0)$ and $\cost(\mathcal{OPT}_0^t) \geq (1 - \frac{\epsilon}{6}) \cdot \cost(\mathcal{OPT}_0)$.
\end{obs}

\begin{proof} 
By \Cref{obsinsert}, $\cost(\mathcal{X}_0^t) \leq \cost(\mathcal{X}_0) + \frac{\epsilon}{6 \alpha} \cdot \cost(\mathcal{B}_0) \leq \cost(\mathcal{X}_0) + \frac{\epsilon}{6} \cdot \cost(\mathcal{OPT}_0)$. Similarly, by \Cref{obsinsert}, $\cost(\mathcal{OPT}_0^t) \geq \cost(\mathcal{OPT}_0) - \frac{\epsilon}{6 \alpha} \cdot \cost(\mathcal{B}_0) \geq \cost(\mathcal{OPT}_0) - \frac{\epsilon}{6} \cdot \cost(\mathcal{OPT}_0)$.
\end{proof}

\begin{obs} \label{w3}
The approximation factor $t$ update steps after the beginning of the initial interval and until it ends is $(1 + \frac{\epsilon}{3}) \cdot \alpha$.
\end{obs}

\begin{proof}
By \Cref{w1} we get that:
\[
\begin{array}{rcl}
\displaystyle
\frac{\cost(\mathcal{X}_0^t)}{\cost(\mathcal{OPT}_0^t)} 
& \leq & 
\displaystyle \frac{\cost(\mathcal{X}_0) + \frac{\epsilon}{6} \cdot \cost(\mathcal{OPT}_0)}{(1 - \frac{\epsilon}{6}) \cdot \cost(\mathcal{OPT}_0)} \\[4mm]
& \leq & 
\displaystyle \frac{(\alpha + \frac{\epsilon}{6}) \cdot \cost(\mathcal{OPT}_0)}{(1 - \frac{\epsilon}{6}) \cdot \cost(\mathcal{OPT}_0)} \\[4mm]
& = & 
\displaystyle \frac{\alpha + \frac{\epsilon}{6}}{1 - \frac{\epsilon}{6}} \\[4mm]
& \leq & 
(1 + \frac{\epsilon}{3}) \cdot \alpha,
\end{array}
\]

\noindent where the last inequality holds for any $\epsilon < 1$ and $\alpha \geq 2$. 
\end{proof}

\noindent We have shown that both items of \Cref{wapprox} hold for the base case ($i = 0$). For the induction step, we begin with the first item:

\begin{obs} \label{apsecond}
The approximation factor $t$ update steps after the beginning of the $i$th interval and until it ends is $(2 + \epsilon) \cdot \alpha$.
\end{obs}

\begin{proof}
\[
\begin{array}{rcl}
\displaystyle
\frac{\cost(\mathcal{X}_i^t)}{\cost(\mathcal{OPT}_i^t)} 
& \leq & 
\displaystyle \frac{\cost(\mathcal{X}_i) + \cost(\mathcal{B}_i) + \frac{\epsilon}{6 \alpha} \cdot \max \{\cost(\mathcal{X}_i),\cost(\mathcal{B}_i)\}}{\cost(\mathcal{OPT}_i) - \frac{\epsilon}{6 \alpha} \cdot \max \{\cost(\mathcal{X}_i),\cost(\mathcal{B}_i)\}} \\[4mm]
& \leq & 
\displaystyle \frac{(1 + \frac{\epsilon}{3}) \cdot \alpha + \alpha + \frac{\epsilon}{6} \cdot (1 + \frac{\epsilon}{3})}{1 - \frac{\epsilon}{6} \cdot (1 + \frac{\epsilon}{3})} \\[4mm]
& \leq & 
(2 + \epsilon) \cdot \alpha,
\end{array}
\]

\noindent where the first inequality is by \Cref{obsinsert}, the second by the induction hypothesis (\Cref{wapprox} (2)), and the last holds for any $\epsilon \leq 0.7$ and $\alpha \geq 2$.
\end{proof}

\begin{obs} \label{cyc}
At the beginning of the $(i+1)$th interval, the approximation factor is $(1 + \frac{\epsilon}{3}) \cdot \alpha$.
\end{obs}

\begin{proof}
\[
\begin{array}{rcl}
\displaystyle
\frac{\cost(\mathcal{X}_{i+1})}{\cost(\mathcal{OPT}_{i+1})} 
& \leq & 
\displaystyle \frac{\cost(\mathcal{B}_i) + \frac{\epsilon}{6 \alpha} \cdot \max \{\cost(\mathcal{X}_i),\cost(\mathcal{B}_i)\}}{\cost(\mathcal{OPT}_i) - \frac{\epsilon}{6 \alpha} \cdot \max \{\cost(\mathcal{X}_i),\cost(\mathcal{B}_i)\}} \\[4mm]
& \leq & 
\displaystyle \frac{\alpha + \frac{\epsilon}{6} \cdot (1 + \frac{\epsilon}{3})}{1 - \frac{\epsilon}{6} \cdot (1 + \frac{\epsilon}{3})} \\[4mm]
& \leq & 
(1 + \frac{\epsilon}{3}) \cdot \alpha,
\end{array}
\]

\noindent where the first inequality is by \Cref{obsinsert}, the second by the induction hypothesis (\Cref{wapprox} (2)), and the last holds for any $\epsilon \leq 0.5$ and $\alpha \geq 2$.
\end{proof}

\noindent We have shown that both items of \Cref{wapprox} hold for the induction step, concluding its proof. This completes the proof of \Cref{warmupthm}.

%% file: Sec4.tex
\section{Worst-Case Recourse: $O\brac{\frac{\log n \cdot C}{\epsilon}}$ to $O\brac{\frac{C}{\epsilon}}$ (Greedy)} \label{secim}

We begin by recording the main greedy-based result obtained by \cite{10756164}:

\begin{thm} [\cite{10756164} - Greedy-Based] \label{suzgreedy}
For any set system $(\univ, \sets)$ that undergoes a sequence of element insertions and deletions, where the frequency is bounded by $f$, and for any $\epsilon \in (0, \frac{1}{4})$, there is a deterministic algorithm that maintains a $((1+\epsilon)\ln n)$-approximate SC in $O\brac{\frac{f\log n}{\epsilon^2}}$  worst-case update time.
\end{thm}

\noindent We can feed this algorithm into our black-box transformation presented in \Cref{techsec} to obtain a $((2+\epsilon) \ln n)$-approximate SC in $O\brac{\frac{f \cdot \log n}{\epsilon^2} + \frac{\log n \cdot C}{\epsilon}}$ worst-case update time and with a worst-case recourse of $O\brac{\frac{\log n \cdot C}{\epsilon}}$. Our goal is to remove the $\log n$ factor from the recourse bound. The basic idea is to increase the intervals' lengths as defined in \Cref{techsec} by a factor of $\ln n$ to reduce the recourse, and argue---using the robustness established in \Cref{dynamicrob}---that the approximation factor remains in check. However, the situation is much more subtle, since the robustness of a SC solution depends on its cost, while the costs of the source and target solutions at the start and end of each interval may differ by a factor of $\ln n$. This subtlety requires a considerably more intricate analysis; see \Cref{algapp} for details.

\subsection{Algorithm Description} \label{algdes}

We use the same algorithm as in \Cref{techsec}, with the only modification being that the interval lengths are increased by a factor of $\alpha = \ln n$. Namely, the length of the $i$th interval will be $\frac{\epsilon}{6} \cdot \max \{\cost(\mathcal{X}_i),\cost(\mathcal{B}_i)\}$ update steps, and each phase will be of length $\frac{\epsilon}{12} \cdot \max \{\cost(\mathcal{X}_i),\cost(\mathcal{B}_i)\}$ update steps. The initial interval will simply be of length $\frac{\epsilon}{6} \cdot \cost(\mathcal{B}_0)$.

\begin{obs} [Recourse] \label{obsrecourse}
The worst-case recourse is $\frac{12C}{\epsilon} + 1$.
\end{obs}

\begin{proof}
In each phase in the $i$th interval we gradually add/remove $\leq C \cdot \max \{\cost(\mathcal{X}_i),\cost(\mathcal{B}_i)\}$ sets. This is done in up to $\frac{\epsilon}{12} \cdot \max \{\cost(\mathcal{X}_i),\cost(\mathcal{B}_i)\}$ update steps. We add one more set per insertion.
\end{proof}

\begin{obs} [Update Time] \label{obstime}
The worst-case update time is $O\brac{\frac{f \cdot \log n}{\epsilon^2} + \frac{C}{\epsilon}}$.
\end{obs}

\begin{proof}
Clearly the bottleneck is either running \cite{10756164} in the background or the recourse.
\end{proof}

\begin{obs} [Legal Solution] \label{obslegal}
At all times the output is a legal SC solution (covers all elements).
\end{obs}

\begin{proof}
In the beginning of the $i$th interval we have two legal solutions, $\mathcal{X}_i$ and $\mathcal{B}_i$. Throughout the $i$th interval the output solution always fully contains at least one of them, plus a set containing each inserted element.
\end{proof}

\subsection{Approximation Factor} \label{algapp}

Throughout the $i$th interval our output solution is contained in $(\mathcal{B}_i \cup \mathcal{N}_i \cup \mathcal{X}_i)$. Thus, we would like to give an upper bound to the approximation factor of this legal solution. We will split to $\mathcal{X}_i$, the \emph{source} solution and $(\mathcal{B}_i \cup \mathcal{N}_i)$, the \emph{target} solution. Instead of explicitly analyzing the approximation factor of each one, each interval $i$ will have an \emph{anchor}, which is a solution given by \cite{10756164}, specifically $\mathcal{B}_{i^*}$ for some $i^* \leq i$. Loosely, our goal is to show that the cost of the source/target solution throughout the $i$th interval is no larger than the cost of some ``theoretical'' solution, that we do not actually maintain, that has a good approximation factor throughout the $i$th interval, which is the naive extension of the anchor throughout the $i$th interval. Thus, we must show three things:

\begin{enumerate}
    \item The cost of the source solution is no larger than the cost of the naively maintained anchor.
    \item The cost of the target solution is no larger than the cost of the naively maintained anchor.
    \item The naively maintained anchor has a good approximation factor throughout the $i$th interval, which will be done by applying \Cref{robusthf} on it.
\end{enumerate}

\noindent We emphasize that we do not actually maintain the naively extended anchor, it is used solely for the analysis. We first argue that the easy case is if $\cost(\mathcal{X}_i) \leq \cost(\mathcal{B}_i)$. Indeed, since the ``fresh'' SC, $\mathcal{B}_i$, is the dominant one, we can simply take $i^* = i$, meaning $\mathcal{B}_i$ is the anchor. Clearly the naively maintained anchor is a $((1 + O(\epsilon)) \cdot \ln n)$-approximate SC throughout the $i$th interval, since the length of the $i$th interval is $\frac{\epsilon}{6} \cdot \cost (\mathcal{B}_i)$. Thus we can apply \Cref{robusthf} with $\delta = \frac{\epsilon}{6}$. Moreover, the cost of the target solution is no larger than the cost of the naively maintained anchor throughout the $i$th interval trivially because it \emph{is} the naively maintained anchor throughout the $i$th interval. The cost of the source solution is no larger than the cost of the naively maintained anchor throughout the $i$th interval since $\cost(\mathcal{X}_i) \leq \cost(\mathcal{B}_i)$, and throughout the interval the source solution does not change and the cost of the naively maintained anchor can only rise. Since throughout the $i$th interval our output solution is contained in $(\mathcal{B}_i \cup \mathcal{N}_i \cup \mathcal{X}_i)$, and $\cost(\mathcal{B}_i \cup \mathcal{N}_i \cup \mathcal{X}_i) \leq \cost(\mathcal{B}_i \cup \mathcal{N}_i) + \cost(\mathcal{X}_i)$, we conclude with the following corollary:

\begin{cor} \label{wrap1}
If $\cost(\mathcal{X}_i) \leq \cost(\mathcal{B}_i)$, then throughout the $i$th interval the output solution gives an approximation ratio of $(2 + O(\epsilon)) \cdot \ln n$.
\end{cor}

\noindent If $\cost(\mathcal{X}_i) > \cost(\mathcal{B}_i)$, the anchor of the $i$th interval will be $\mathcal{B}_{i^*}$ where $i^*$ is the last interval $< i$ such that $\cost(\mathcal{X}_{i^*}) \leq 3 \cdot \cost(\mathcal{B}_{i^*})$. Note that $i^*$ is well defined since $\cost(\mathcal{X}_{0}) = \cost(\mathcal{B}_{0})$.

\begin{cl} \label{equa9}
For any $i^* + 1 \leq j \leq i$, we have that:
$$
\cost(\mathcal{X}_j) \leq \displaystyle \left(\frac{2 + \epsilon}{6}\right)^{\,j - i^* - 1} \cdot \left(1 + \frac{\epsilon}{2}\right) \cdot \cost(\mathcal{B}_{i^*}).
$$
\end{cl}

\begin{proof}
For any $i^* + 1 \leq j < i$, we have that $\cost(\mathcal{X}_{j}) > 3 \cdot \cost(\mathcal{B}_{j})$. Note that for any $i^* + 2 \leq j \leq i$, we also have:
\begin{equation} \label{eq6}
\begin{array}{rcl}
\cost(\mathcal{X}_j)
& \leq & \cost(\mathcal{B}_{j-1}) + \cost(\mathcal{N}_{j-1}) < \displaystyle \frac{1}{3} \cdot \cost(\mathcal{X}_{j-1}) + \frac{\epsilon}{6} \cdot \cost(\mathcal{X}_{j-1}) \\[2mm]
& = & \displaystyle \left(\frac{2 + \epsilon}{6}\right) \cdot \cost(\mathcal{X}_{j-1}),
\end{array}
\end{equation}

\noindent and specifically for $j = i^* + 1$ we have:
\begin{equation} \label{eq7}
\begin{array}{rcl}
\cost(\mathcal{X}_{i^* + 1})
& \leq & \cost(\mathcal{B}_{i^*}) + \cost(\mathcal{N}_{i^*}) \leq  \cost(\mathcal{B}_{i^*}) + 3 \cdot \frac{\epsilon}{6} \cdot \cost(\mathcal{B}_{i^*}) \\[2mm]
& = & \displaystyle \left(1 + \frac{\epsilon}{2}\right) \cdot \cost(\mathcal{B}_{i^*}).
\end{array}
\end{equation}

\noindent Thus, for any $i^* + 1 \leq j \leq i$ we have:
\begin{equation} \label{eq8}
\begin{array}{rcl}
\cost(\mathcal{X}_j)
& \leq & \displaystyle \left(\frac{2 + \epsilon}{6}\right)^{\,j - i^* - 1} \cdot \cost(\mathcal{X}_{i^* + 1}) \\[4mm]
& \leq & \displaystyle \left(\frac{2 + \epsilon}{6}\right)^{\,j - i^* - 1} \cdot \left(1 + \frac{\epsilon}{2}\right) \cdot \cost(\mathcal{B}_{i^*}).
\end{array}
\end{equation}  
\end{proof}

\begin{cl} \label{anchor}
The number of update steps from the beginning of the $(i^*)$th interval to the end of the $i$th interval is $\leq \frac{9}{10} \cdot \epsilon \cdot \cost(\mathcal{B}_{i^*})$.
\end{cl}

\begin{proof}
The total number of update steps by \Cref{equa9}:
\begin{equation} \label{eq9}
\begin{array}{rcl}
\displaystyle \sum_{j = i^*}^{i} |\mathcal{N}_j|
& \leq & \displaystyle \frac{\epsilon}{2} \cost(\mathcal{B}_{i^*})
+ \frac{\epsilon}{6} \sum_{j = i^* + 1}^{i} \cost(\mathcal{X}_j) \\[2mm]
& \leq & \displaystyle \frac{\epsilon}{2} \cost(\mathcal{B}_{i^*})
+ \frac{\epsilon}{6} \left(1 + \frac{\epsilon}{2}\right) \cost(\mathcal{B}_{i^*})
\sum_{j = i^* + 1}^{i} \left(\frac{2 + \epsilon}{6}\right)^{j - i^* - 1}.
\end{array}
\end{equation}

\noindent By bounding the geometric sum, we obtain: 
\begin{equation} \label{eq10}
\sum_{j = i^*}^i |\mathcal{N}_j| \leq \cost(\mathcal{B}_{i^*}) \cdot \brac{\frac{\epsilon}{2} + \frac{\epsilon}{6} \cdot \brac{1 + \frac{\epsilon}{2}} \cdot \brac{\frac{6}{4-\epsilon}}} \leq \frac{9}{10} \cdot \epsilon \cdot \cost(\mathcal{B}_{i^*}),
\end{equation}

\noindent where the last inequality holds for any $\epsilon \leq \frac{1}{4}$, and the claim follows.    
\end{proof}

\begin{cor} \label{anc2}
The naively maintained anchor is a $((1 + O(\epsilon)) \cdot \ln n)$-approximate SC throughout the $i$th interval.
\end{cor}

\begin{proof}
By \Cref{anchor}, we naively maintain our anchor $\mathcal{B}_{i^*}$ for up to $\frac{9}{10} \cdot \epsilon \cdot \cost(\mathcal{B}_{i^*})$ update steps, so we can apply \Cref{robusthf} on it with $\delta = \frac{9}{10} \cdot \epsilon$ and conclude the proof.
\end{proof}

\begin{obs} \label{tar2}
The cost of the target solution is no larger than $(1 + O(\epsilon)) \cdot \cost(\mathcal{B}_{i^*})$.
\end{obs}

\begin{proof}
We have that:
$$\cost(\mathcal{B}_i \cup \mathcal{N}_i) \leq \cost(\mathcal{B}_{i}) + \cost(\mathcal{N}_{i}) < \cost(\mathcal{X}_{i}) + \frac{\epsilon}{6} \cost(\mathcal{X}_{i}) = \brac{1 + \frac{\epsilon}{6}} \cost(\mathcal{X}_{i}).$$    

\noindent By plugging in $j=i$ in \Cref{equa9}, we obtain that:
$$\cost(\mathcal{B}_i \cup \mathcal{N}_i) < \brac{1 + \frac{\epsilon}{6}}\brac{1 + \frac{\epsilon}{2}} \cdot \cost(\mathcal{B}_{i^*}),$$
\noindent and so the claim holds.
\end{proof}

\begin{obs} \label{sou2}
The cost of the source solution is no larger than $(1 + O(\epsilon)) \cdot \cost(\mathcal{B}_{i^*})$.
\end{obs}

\begin{proof}
The observation follows immediately by plugging in $j=i$ in \Cref{equa9}.
\end{proof}

\begin{cor} \label{wrap2}
If $\cost(\mathcal{X}_i) > \cost(\mathcal{B}_i)$, then throughout the $i$th interval the output solution gives an approximation ratio of $(2 + O(\epsilon)) \cdot \ln n$.
\end{cor}

\begin{proof}
Throughout the $i$th interval our output solution is contained in $(\mathcal{B}_i \cup \mathcal{N}_i \cup \mathcal{X}_i)$. Since $\cost(\mathcal{B}_i \cup \mathcal{N}_i \cup \mathcal{X}_i) \leq \cost(\mathcal{B}_i \cup \mathcal{N}_i) + \cost(\mathcal{X}_i)$, the corollary follows immediately from \Cref{anc2}, \Cref{tar2} and \Cref{sou2}.
\end{proof}

\begin{cor} \label{wrap}
Throughout the entire update sequence, the output solution gives an approximation ratio of $(2 + \epsilon) \cdot \ln n$.
\end{cor}

\begin{proof}
The corollary holds by combining \Cref{wrap1} and \Cref{wrap2} (and by scaling $\epsilon$).
\end{proof}

\begin{cor}
\Cref{hf} holds by \Cref{obsrecourse}, \Cref{obstime}, \Cref{obslegal} and \Cref{wrap}.
\end{cor}

%% file: main.bbl
\begin{thebibliography}{MSVW16}

\bibitem[AAG{\etalchar{+}}19]{abboud2019dynamic}
Amir Abboud, Raghavendra Addanki, Fabrizio Grandoni, Debmalya Panigrahi, and Barna Saha.
\newblock Dynamic set cover: improved algorithms and lower bounds.
\newblock In {\em Proceedings of the 51st Annual ACM SIGACT Symposium on Theory of Computing}, pages 114--125, 2019.

\bibitem[ADJ20]{r2}
Spyros Angelopoulos, Christoph D\"{u}rr, and Shendan Jin.
\newblock Online maximum matching with recourse.
\newblock {\em J. Comb. Optim.}, 40(4):974–1007, November 2020.

\bibitem[AOSS18a]{r6}
Sepehr Assadi, Krzysztof Onak, Baruch Schieber, and Shay Solomon.
\newblock Fully dynamic maximal independent set with sublinear update time.
\newblock In {\em Proceedings of the 50th Annual ACM SIGACT Symposium on Theory of Computing}, STOC 2018, page 815–826, New York, NY, USA, 2018. Association for Computing Machinery.

\bibitem[AOSS18b]{10.1145/3188745.3188922}
Sepehr Assadi, Krzysztof Onak, Baruch Schieber, and Shay Solomon.
\newblock Fully dynamic maximal independent set with sublinear update time.
\newblock In {\em Proceedings of the 50th Annual ACM SIGACT Symposium on Theory of Computing}, STOC 2018, page 815–826, New York, NY, USA, 2018. Association for Computing Machinery.

\bibitem[AS21]{assadi2021fully}
Sepehr Assadi and Shay Solomon.
\newblock Fully dynamic set cover via hypergraph maximal matching: An optimal approximation through a local approach.
\newblock In {\em 29th Annual European Symposium on Algorithms (ESA 2021)}. Schloss Dagstuhl-Leibniz-Zentrum f{\"u}r Informatik, 2021.

\bibitem[BCH17]{bhattacharya2017deterministic}
Sayan Bhattacharya, Deeparnab Chakrabarty, and Monika Henzinger.
\newblock Deterministic fully dynamic approximate vertex cover and fractional matching in {O}(1) amortized update time.
\newblock In {\em International Conference on Integer Programming and Combinatorial Optimization}, pages 86--98. Springer, 2017.

\bibitem[BDH{\etalchar{+}}19]{8948654}
Soheil Behnezhad, Mahsa Derakhshan, MohammadTaghi Hajiaghayi, Cliff Stein, and Madhu Sudan.
\newblock Fully dynamic maximal independent set with polylogarithmic update time.
\newblock In {\em 2019 IEEE 60th Annual Symposium on Foundations of Computer Science (FOCS)}, pages 382--405, 2019.

\bibitem[BGK{\etalchar{+}}15]{r9}
Nikhil Bansal, Anupam Gupta, Ravishankar Krishnaswamy, Kirk Pruhs, Kevin Schewior, and Cliff Stein.
\newblock {A 2-Competitive Algorithm For Online Convex Optimization With Switching Costs}.
\newblock In Naveen Garg, Klaus Jansen, Anup Rao, and Jos\'{e} D.~P. Rolim, editors, {\em Approximation, Randomization, and Combinatorial Optimization. Algorithms and Techniques (APPROX/RANDOM 2015)}, volume~40 of {\em Leibniz International Proceedings in Informatics (LIPIcs)}, pages 96--109, Dagstuhl, Germany, 2015. Schloss Dagstuhl -- Leibniz-Zentrum f{\"u}r Informatik.

\bibitem[BGS11]{6108199}
Surender Baswana, Manoj Gupta, and Sandeep Sen.
\newblock Fully dynamic maximal matching in {O}(logn) update time.
\newblock In {\em 2011 IEEE 52nd Annual Symposium on Foundations of Computer Science}, pages 383--392, 2011.

\bibitem[BHI15]{bhattacharya2015design}
Sayan Bhattacharya, Monika Henzinger, and Giuseppe~F Italiano.
\newblock Design of dynamic algorithms via primal-dual method.
\newblock In {\em International Colloquium on Automata, Languages, and Programming}, pages 206--218. Springer, 2015.

\bibitem[BHN19]{bhattacharya2019new}
Sayan Bhattacharya, Monika Henzinger, and Danupon Nanongkai.
\newblock A new deterministic algorithm for dynamic set cover.
\newblock In {\em 2019 IEEE 60th Annual Symposium on Foundations of Computer Science (FOCS)}, pages 406--423. IEEE, 2019.

\bibitem[BHNW21]{bhattacharya2021dynamic}
Sayan Bhattacharya, Monika Henzinger, Danupon Nanongkai, and Xiaowei Wu.
\newblock Dynamic set cover: Improved amortized and worst-case update time.
\newblock In {\em Proceedings of the 2021 ACM-SIAM Symposium on Discrete Algorithms (SODA)}, pages 2537--2549. SIAM, 2021.

\bibitem[BHR19]{r15}
Aaron Bernstein, Jacob Holm, and Eva Rotenberg.
\newblock Online bipartite matching with amortized {O}$(\log^2 n)$ replacements.
\newblock {\em J. ACM}, 66(5), September 2019.

\bibitem[BK19]{inbook123}
Sayan Bhattacharya and Janardhan Kulkarni.
\newblock {\em Deterministically Maintaining a $(2+\epsilon)$-Approximate Minimum Vertex Cover in {O}$(1/ \epsilon ^2 )$ Amortized Update Time}, pages 1872--1885.
\newblock 01 2019.

\bibitem[BKP{\etalchar{+}}17]{r16}
Aaron Bernstein, Tsvi Kopelowitz, Seth Pettie, Ely Porat, and Clifford Stein.
\newblock {Simultaneously Load Balancing for Every p-norm, With Reassignments}.
\newblock In Christos~H. Papadimitriou, editor, {\em 8th Innovations in Theoretical Computer Science Conference (ITCS 2017)}, volume~67 of {\em Leibniz International Proceedings in Informatics (LIPIcs)}, pages 51:1--51:14, Dagstuhl, Germany, 2017. Schloss Dagstuhl -- Leibniz-Zentrum f{\"u}r Informatik.

\bibitem[BLSZ18]{r28}
Bartłomiej Bosek, Dariusz Leniowski, Piotr Sankowski, and Anna Zych.
\newblock {\em A Tight Bound for Shortest Augmenting Paths on Trees}, pages 201--216.
\newblock 03 2018.

\bibitem[BSZ23]{bukov2023nearly}
Anton Bukov, Shay Solomon, and Tianyi Zhang.
\newblock Nearly optimal dynamic set cover: Breaking the quadratic-in-$ f $ time barrier.
\newblock {\em arXiv preprint arXiv:2308.00793}, 2023.

\bibitem[CDKL09]{r32}
K.~Chaudhuri, C.~Daskalakis, R.~D. Kleinberg, and H.~Lin.
\newblock Online bipartite perfect matching with augmentations.
\newblock In {\em IEEE INFOCOM 2009}, pages 1044--1052, 2009.

\bibitem[CHHK16]{r30}
Keren Censor-Hillel, Elad Haramaty, and Zohar Karnin.
\newblock Optimal dynamic distributed mis.
\newblock In {\em Proceedings of the 2016 ACM Symposium on Principles of Distributed Computing}, PODC '16, page 217–226, New York, NY, USA, 2016. Association for Computing Machinery.

\bibitem[CZ19]{8948632}
Shiri Chechik and Tianyi Zhang.
\newblock Fully dynamic maximal independent set in expected poly-log update time.
\newblock In {\em 2019 IEEE 60th Annual Symposium on Foundations of Computer Science (FOCS)}, pages 370--381, 2019.

\bibitem[DS14]{dinur2014analytical}
Irit Dinur and David Steurer.
\newblock Analytical approach to parallel repetition.
\newblock In {\em Proceedings of the forty-sixth annual ACM symposium on Theory of computing}, pages 624--633, 2014.

\bibitem[GGK13]{r36}
Albert Gu, Anupam Gupta, and Amit Kumar.
\newblock The power of deferral: maintaining a constant-competitive steiner tree online.
\newblock In {\em Proceedings of the Forty-Fifth Annual ACM Symposium on Theory of Computing}, STOC '13, page 525–534, New York, NY, USA, 2013. Association for Computing Machinery.

\bibitem[GK14]{r38}
Anupam Gupta and Amit Kumar.
\newblock Online steiner tree with deletions.
\newblock In {\em Proceedings of the Twenty-Fifth Annual ACM-SIAM Symposium on Discrete Algorithms}, SODA '14, page 455–467, USA, 2014. Society for Industrial and Applied Mathematics.

\bibitem[GKKP17]{gupta2017online}
Anupam Gupta, Ravishankar Krishnaswamy, Amit Kumar, and Debmalya Panigrahi.
\newblock Online and dynamic algorithms for set cover.
\newblock In {\em Proceedings of the 49th Annual ACM SIGACT Symposium on Theory of Computing}, pages 537--550, 2017.

\bibitem[GKKV95]{r35}
Edward~F. Grove, Ming-Yang Kao, P.~Krishnan, and Jeffrey~Scott Vitter.
\newblock Online perfect matching and mobile computing.
\newblock In {\em Proceedings of the 4th International Workshop on Algorithms and Data Structures}, WADS '95, page 194–205, Berlin, Heidelberg, 1995. Springer-Verlag.

\bibitem[GKS]{r39}
Anupam Gupta, Amit Kumar, and Cliff Stein.
\newblock {\em Maintaining Assignments Online: Matching, Scheduling, and Flows}, pages 468--479.

\bibitem[GP13]{6686191}
Manoj Gupta and Richard Peng.
\newblock { Fully Dynamic $(1+ \epsilon)$-Approximate Matchings }.
\newblock In {\em 2013 IEEE 54th Annual Symposium on Foundations of Computer Science (FOCS)}, pages 548--557, Los Alamitos, CA, USA, October 2013. IEEE Computer Society.

\bibitem[Heu13]{Heuvel_2013}
Jan van~den Heuvel.
\newblock {\em The complexity of change}, page 127–160.
\newblock London Mathematical Society Lecture Note Series. Cambridge University Press, 2013.

\bibitem[HIPS19]{hjuler_et_al:LIPIcs.STACS.2019.35}
Niklas Hjuler, Giuseppe~F. Italiano, Nikos Parotsidis, and David Saulpic.
\newblock {Dominating Sets and Connected Dominating Sets in Dynamic Graphs}.
\newblock In Rolf Niedermeier and Christophe Paul, editors, {\em 36th International Symposium on Theoretical Aspects of Computer Science (STACS 2019)}, volume 126 of {\em Leibniz International Proceedings in Informatics (LIPIcs)}, pages 35:1--35:17, Dagstuhl, Germany, 2019. Schloss Dagstuhl -- Leibniz-Zentrum f{\"u}r Informatik.

\bibitem[IDH{\etalchar{+}}11]{ito}
Takehiro Ito, Erik~D. Demaine, Nicholas J.~A. Harvey, Christos~H. Papadimitriou, Martha Sideri, Ryuhei Uehara, and Yushi Uno.
\newblock On the complexity of reconfiguration problems.
\newblock {\em Theor. Comput. Sci.}, 412(12–14):1054–1065, March 2011.

\bibitem[KR08]{khot2008vertex}
Subhash Khot and Oded Regev.
\newblock Vertex cover might be hard to approximate to within $2-\varepsilon$.
\newblock {\em Journal of Computer and System Sciences}, 74(3):335--349, 2008.

\bibitem[MSVW16]{r51}
Nicole Megow, Martin Skutella, Jos\'{e} Verschae, and Andreas Wiese.
\newblock The power of recourse for online mst and tsp.
\newblock {\em SIAM Journal on Computing}, 45(3):859--880, 2016.

\bibitem[Nis18]{nishimura}
Naomi Nishimura.
\newblock Introduction to reconfiguration.
\newblock {\em Algorithms}, 11(4), 2018.

\bibitem[NS15]{10.1145/2700206}
Ofer Neiman and Shay Solomon.
\newblock Simple deterministic algorithms for fully dynamic maximal matching.
\newblock {\em ACM Trans. Algorithms}, 12(1), November 2015.

\bibitem[OR10]{10.1145/1806689.1806753}
Krzysztof Onak and Ronitt Rubinfeld.
\newblock Maintaining a large matching and a small vertex cover.
\newblock In {\em Proceedings of the Forty-Second ACM Symposium on Theory of Computing}, STOC '10, page 457–464, New York, NY, USA, 2010. Association for Computing Machinery.

\bibitem[PS]{doi:10.1137/1.9781611974331.ch51}
David Peleg and Shay Solomon.
\newblock {\em Dynamic $(1 + \epsilon)$-Approximate Matchings: A Density-Sensitive Approach}, pages 712--729.

\bibitem[Sol16]{7782946}
Shay Solomon.
\newblock { Fully Dynamic Maximal Matching in Constant Update Time }.
\newblock In {\em 2016 IEEE 57th Annual Symposium on Foundations of Computer Science (FOCS)}, pages 325--334, Los Alamitos, CA, USA, October 2016. IEEE Computer Society.

\bibitem[Sol18]{Solomon18}
Shay Solomon.
\newblock Local algorithms for bounded degree sparsifiers in sparse graphs.
\newblock In Anna~R. Karlin, editor, {\em 9th Innovations in Theoretical Computer Science Conference, {ITCS} 2018, January 11-14, 2018, Cambridge, MA, {USA}}, volume~94 of {\em LIPIcs}, pages 52:1--52:19. Schloss Dagstuhl - Leibniz-Zentrum f{\"{u}}r Informatik, 2018.

\bibitem[SS21]{ShayNoam}
Noam Solomon and Shay Solomon.
\newblock {A Generalized Matching Reconfiguration Problem}.
\newblock In James~R. Lee, editor, {\em 12th Innovations in Theoretical Computer Science Conference (ITCS 2021)}, volume 185 of {\em Leibniz International Proceedings in Informatics (LIPIcs)}, pages 57:1--57:20, Dagstuhl, Germany, 2021. Schloss Dagstuhl -- Leibniz-Zentrum f{\"u}r Informatik.

\bibitem[SU23]{solomon2023dynamic}
Shay Solomon and Amitai Uzrad.
\newblock {Dynamic $((1+\epsilon)\ln n)$-Approximation Algorithms for Minimum Set Cover and Dominating Set}.
\newblock In {\em Proceedings of the 55th Annual ACM Symposium on Theory of Computing}, pages 1187--1200, 2023.

\bibitem[SUZ24]{10756164}
Shay Solomon, Amitai Uzrad, and Tianyi Zhang.
\newblock A lossless deamortization for dynamic greedy set cover.
\newblock In {\em 2024 IEEE 65th Annual Symposium on Foundations of Computer Science (FOCS)}, pages 264--290, 2024.

\bibitem[WS11]{williamson2011design}
David~P Williamson and David~B Shmoys.
\newblock {\em The design of approximation algorithms}.
\newblock Cambridge university press, 2011.

    
\end{thebibliography}
